\documentclass[sigplan,10pt,preprint]{acmart}
\renewcommand\footnotetextcopyrightpermission[1]{} 

\usepackage[normalem]{ulem}

\usepackage{textcomp}
\usepackage{adjustbox}
\usepackage{xspace}
\usepackage{amsmath}
\usepackage{amsthm}
\usepackage{float}
\usepackage{algpseudocode}
\usepackage{graphicx}
\usepackage[skip=2pt]{caption}
\DeclareCaptionType{copyrightbox}
\usepackage{wrapfig}
\usepackage{microtype}
\usepackage{xcolor}
\usepackage{multirow}
\usepackage{multicol}
\usepackage{balance}
\usepackage{nowidow}
\usepackage{import}
\usepackage{mathtools}
\usepackage{listings}
\usepackage{url}
\usepackage{hhline}
\usepackage{tabularx}
\usepackage{subfig}
\usepackage[mathscr]{euscript}
 \let\mathscr\relax
\usepackage{bigints}
\usepackage{adjustbox}
\usepackage{enumitem}
\usepackage{array}
\usepackage{makecell}
\usepackage{bm}
\newcolumntype{x}[1]{>{\centering\arraybackslash\hspace{0pt}}p{#1}}
\newcolumntype{L}[1]{>{\raggedright\let\newline\\\arraybackslash\hspace{0pt}}m{#1}}
\newcolumntype{C}[1]{>{\centering\let\newline\\\arraybackslash\hspace{0pt}}m{#1}}
\newcolumntype{R}[1]{>{\raggedleft\let\newline\\\arraybackslash\hspace{0pt}}m{#1}}

\usepackage{color, colortbl}
\usepackage{color}

\definecolor{lightgray}{gray}{0.92}
\definecolor{applegreen}{rgb}{0.13, 0.67, 0.8}
\makeatletter 
\renewcommand{\@seccntformat}[1]{\csname the#1\endcsname\ \ }
\makeatother

\usepackage{algorithm}
\usepackage{algorithmicx}

\newcommand{\BigO}[1]{\ensuremath{\operatorname{O}\bigl(#1\bigr)}}

\newcommand{\rfig}[1]{Figure~\ref{#1}}
\newcommand{\rcor}[1]{Corollary.~\ref{#1}}
\newcommand{\rtab}[1]{Table~\ref{#1}}

\newcommand{\rsec}[1]{Section~\ref{#1}}
\newcommand{\rthm}[1]{Theorem~\ref{#1}}

\newcommand{\Naive}[0]{Na\"{i}ve\xspace}
\newcommand{\naively}[0]{na\"{i}vely\xspace}

\algblock{ParFor}{EndParFor}
\algnewcommand\algorithmicparfor{\textbf{par-for}}
\algnewcommand\algorithmicpardo{\textbf{do}}
\algnewcommand\algorithmicendparfor{\textbf{end\ par-for}}
\algrenewtext{ParFor}[1]{\algorithmicparfor\ #1\ \algorithmicpardo}
\algrenewtext{EndParFor}{\algorithmicendparfor}

\algblock{MyFunc}{EndMyFunc}
\algrenewtext{MyFunc}[3]{#1\ \textbf{#2}\ #3}
\algrenewtext{EndMyFunc}{}

\algrenewcommand\algorithmicindent{1.0em}%
\newcommand{\MyComment}[1]{}
\algdef{SE}[DOWHILE]{Do}{DoWhile}{\algorithmicdo}[1]{\algorithmicwhile\ #1}

\algblock{MyWhile}{EndMyWhile}
\algnewcommand\algorithmicmywhile{\textbf{while}}
\algnewcommand\algorithmicmywhiledo{\textbf{do}}
\algnewcommand\algorithmicendmywhile{\textbf{end while}}
\algrenewtext{MyWhile}[1]{\algorithmicmywhile\ #1\ \algorithmicpardo}
\algrenewtext{EndMyWhile}{\algorithmicendmywhile}

\theoremstyle{plain}
\newtheorem{thm}{Theorem}[section]

\newtheorem{cor}{Corollary}[section]

\theoremstyle{definition}

\newcommand{\nomorph}{No PMR}
\newcommand{\naivemorph}{\Naive{} PMR}
\newcommand{\smartmorph}{Cost-Based PMR}

\begin{document}

\title{\textsc{Pattern Morphing} for Efficient Graph Mining}

\author{Kasra Jamshidi}
\affiliation{%
\department{School of Computing Science}
  \institution{Simon Fraser University}
  \state{British Columbia, Canada}
}
\email{kjamshid@cs.sfu.ca}

\author{Keval Vora}
\affiliation{%
\department{School of Computing Science}
  \institution{Simon Fraser University}
  \state{British Columbia, Canada}
}
\email{keval@cs.sfu.ca}

\begin{abstract}
Graph mining applications analyze the structural properties of large graphs, and they do so by finding subgraph isomorphisms, which makes them computationally intensive. Existing graph mining techniques including both custom graph mining applications and general-purpose graph mining systems, develop efficient execution plans to speed up the exploration of the given query patterns that represent subgraph structures of interest. 

In this paper, we step beyond the traditional philosophy of optimizing the execution plans for a given set of patterns, and exploit the sub-structural similarities across different query patterns. We propose \textsc{Pattern Morphing}, a technique that enables structure-aware algebra over patterns to accurately infer the results for a given set of patterns using the results of a completely different set of patterns that are less expensive to compute. Pattern morphing \emph{morphs} (or converts) a given set of query patterns into alternative patterns, while retaining full equivalency. It is a general technique that supports various operations over matches of a pattern beyond just counting (e.g., support calculation, enumeration, etc.), making it widely applicable to various graph mining applications like Motif Counting and Frequent Subgraph Mining. Since pattern morphing mainly transforms query patterns before their exploration starts, it can be easily incorporated in existing general-purpose graph mining systems. We evaluate the effectiveness of pattern morphing by incorporating it in Peregrine, a recent state-of-the-art graph mining system, and show that pattern morphing significantly improves the performance of different graph mining applications. 
\end{abstract}

\maketitle

\begin{figure}[t]
\includegraphics[width=0.99\linewidth]{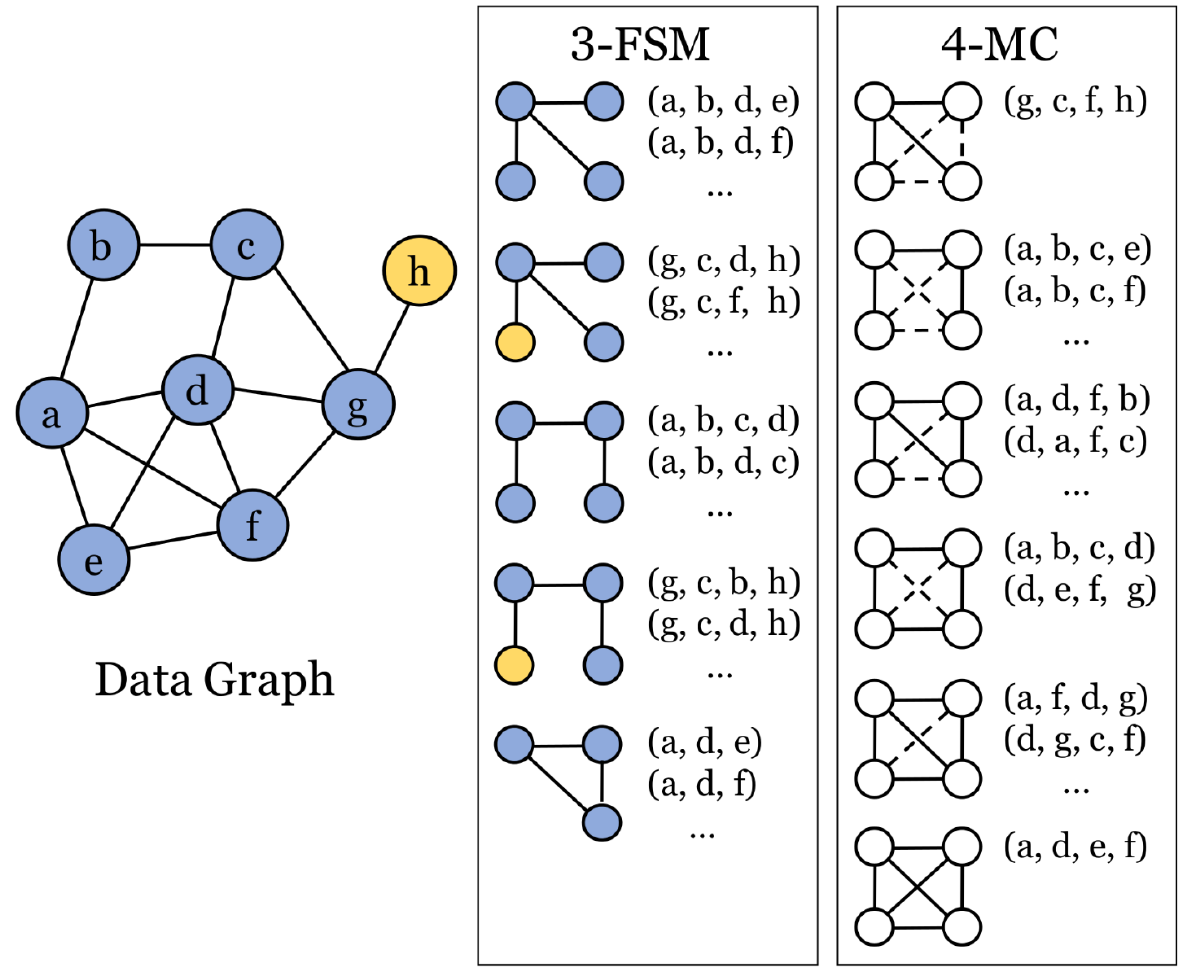}
  \caption{Example of vertex-induced patterns v/s edge-induced patterns in graph mining applications. \\ 3-FSM typically explores (labeled) edge-induced matches, and there are three size-3 pattern topologies (i.e., topologies containing 3 edges) which form five patterns with distinct labelings from the data graph. 4-MC typically explores (unlabeled) vertex-induced matches, and there are six size-4 patterns (i.e., patterns containing 4 vertices).}
\label{fig-backgroundfsmmc}
\end{figure}

\begin{table*}[t]
  \begin{tabular}{c | c c | c c | c c}
    & \multicolumn{2}{c|}{4-Cycle} & \multicolumn{2}{c|}{Chordal 4-Cycle} & \multicolumn{2}{c}{5-Cycle} \\ 
    & Edge-Induced & Vertex-Induced & Edge-Induced & Vertex-Induced & Edge-Induced & Vertex-Induced \\
    & \includegraphics[width=.5in]{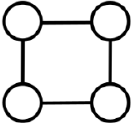} & \includegraphics[width=.5in]{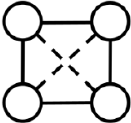} & \includegraphics[width=.5in]{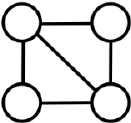} & \includegraphics[width=.5in]{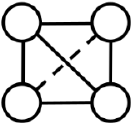} & \includegraphics[width=.5in,angle=90]{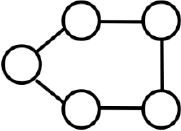} & \includegraphics[width=.5in,angle=90]{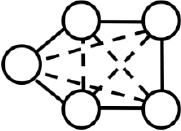} \\
    \midrule
    Mico & 2.09s &  \textbf{1.89s} & \textbf{0.08s} & 3.04s & 258.90s & \textbf{23.56s} \\
    YouTube & \textbf{27.84s} & 31.91s & \textbf{1.62s} & 8.07s & 111.89s &  \textbf{67.03s} \\
    \midrule
  \end{tabular}
  \caption{Execution times (in seconds) for matching various patterns in Peregrine~\cite{peregrine}.}
  \label{tab-performancemotivation}
\end{table*}

\section{Introduction}
With growing interest in analyzing graph-based data, graph mining applications are being widely used across several domains including bioinformatics, computer vision, malware detection, and social network analysis~\cite{fsm-sec,motifs-sna,motifs-sna2,motifs-bio,maxclique-sna,motifsbiology,fsmcompvis,fsmbiology}.

Several general-purpose graph mining systems have been developed in literature~\cite{arabesque,rstream,fractal,gminer,automine,peregrine,pangolin}. These systems provide an efficient runtime that explores the subgraphs of interest along with a high-level programming model to guide the exploration process and perform computations based on the explored subgraphs. Recent systems like Peregrine~\cite{peregrine} take a pattern-aware approach where the graph mining applications are expressed as pattern matching sub-tasks, and the exploration plan directly finds the matching subgraphs, without performing expensive isomorphism and automorphism checks for every match. 

With such an approach, details about how to match the patterns naturally become part of the pattern itself. For instance, \rfig{fig-backgroundfsmmc} shows the difference in the patterns expressed for Frequent Subgraph Mining (FSM) and Motif Counting (MC) applications. Since Frequent Subgraph Mining typically computes frequencies based on \emph{edge-induced} matches (i.e., subgraphs that match all the edges in the pattern), the three size-3 topologies (i.e., topologies containing 3 `edges') are simply expressed by the edges that are present in the pattern. Motif Counting, on the other hand, usually counts the occurrences of \emph{vertex-induced} matches (i.e., subgraphs that match all the edges between vertices in the data graph), and hence the six size-4 patterns (i.e., patterns containing 4 `vertices') contain \emph{anti-edges}~\cite{peregrine} indicating that edges should not be present between those vertices. General-purpose graph mining systems often provide flexibility to use different exploration models for different workloads: for instance, FSM can be performed in edge-induced manner or vertex-induced manner on Peregrine~\cite{peregrine}, and switching between the two requires just a single line change in its FSM program.  

The performance of graph mining systems is mainly dependent on three factors: (1) how many (partial and complete) matches are present in the data graph (i.e., the structure of the data graph); (2) how expensive it is to match the given set of patterns (i.e., the structure of the patterns); and, (3) what to do once the matches have been found (i.e., the application semantics). \rfig{fig-fsmmotifbreakdownmotivation} shows the performance breakdown of 3-FSM and 4-MC on two graphs; as we can see, the majority of time in MC is spent on finding matches and little time is taken for its counting aggregation, whereas the performance of FSM is dominated by performing aggregation (computing pattern supports~\cite{fsm-mni}) over the matches. Hence, it is crucial to reduce both the time taken for matching, and the time taken to perform aggregation in order to enable high performance graph mining.

Existing systems already incorporate efficient pattern matching algorithms; however, they 
limit their analysis to only those set of patterns that are explicitly requested by the mining application. In this work, we aim to exploit the structural relationships between different patterns to accelerate the overall mining task. To understand how we achieve that, first we describe a key insight.

\begin{figure}[t]
  \includegraphics[width=0.5\textwidth]{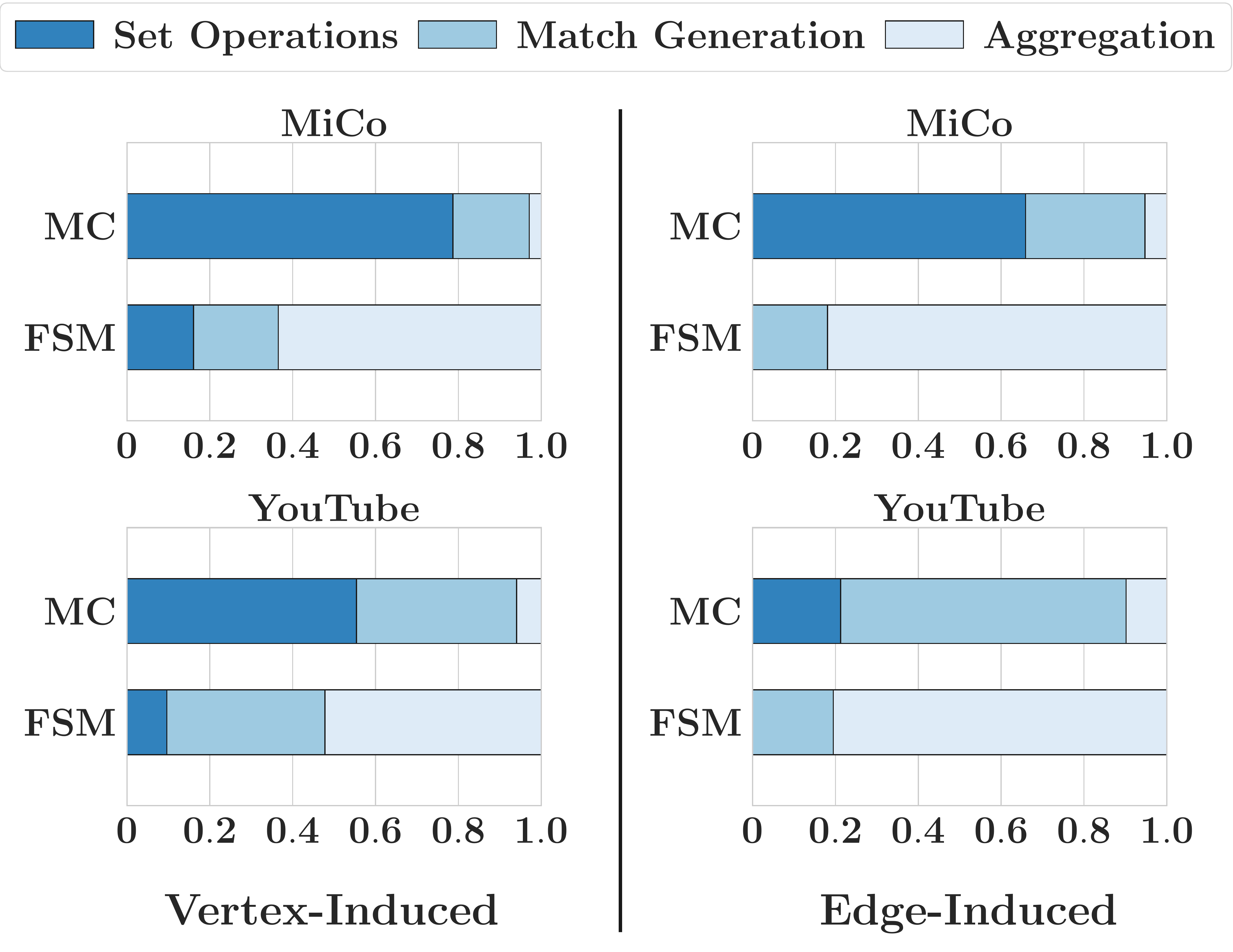}
  \caption{Performance of Frequent Subgraph Mining (FSM) and Motif Counting (MC) on Mico~\cite{mico} and YouTube~\cite{youtube} graphs in Peregrine~\cite{peregrine} with vertex-induced and edge-induced exploration.}
  \label{fig-fsmmotifbreakdownmotivation}
\end{figure}

\paragraph{Observations \& Key Insight.}
\rtab{tab-performancemotivation} shows the time taken for matching three patterns in vertex-induced and edge-induced manner on two data graphs. As we can see, there is no correlation between performance of edge-induced patterns and their respective vertex-induced variants; for instance, matching the edge-induced chordal 4-cycle is faster than matching the vertex-induced chordal 4-cycle, whereas matching the edge-induced 5-cycle is slower than matching its vertex-induced variant. This is because even though anti-edges provide pruning benefits for vertex-induced patterns at each exploration step, enforcing them using set differences can be more expensive than performing set intersections in edge-induced variants. Furthermore, two patterns appearing structurally similar to each other can end up taking drastically different amounts time; for instance, an edge-induced 4-cycle and an edge-induced chordal 4-cycle differ by only one edge, and still matching the former is over 26$\times$ slower than matching the latter.

These drastic variations in performance across structurally similar patterns present opportunities to perform analysis across different patterns, including those that are not explicitly requested by the mining application. This means \emph{the structural similarities between different patterns can be exploited to infer the results for a given set of patterns using the results of a different set of patterns that are relatively inexpensive to compute}.

\paragraph{Pattern Morphing.}
In this paper, we propose \textsc{Pattern Morphing}, a technique that enables structure-aware algebra over patterns to \emph{morph} (or convert) the query patterns into alternative patterns, which can then be used to quickly compute accurate final results for the original query patterns. We first formalize the semantics of pattern morphing to understand how patterns can be morphed into alternative patterns while retaining full equivalency with the original patterns so that correct final results can be guaranteed. Then we develop an automatic pattern morphing engine that reformulates the original graph mining query into a query on the morphed (alternative) patterns, hence exploiting alternative matching plans that deliver different performance. 

Our pattern morphing theory is general, and captures the application semantics in form of operations to be performed based on the matched subgraphs. This not only enables simple counting based applications like Motif Counting, but also supports complex aggregation types like MNI support computations~\cite{fsm-mni} required in Frequent Subgraph Mining.

Since any given set of patterns can be morphed into different but equivalent alternative patterns, selecting the right alternative pattern sets is crucial to deliver high performance. This is addressed by developed a \emph{pattern morphing optimizer} that minimizes the cost of pattern sets to construct the best alternative pattern sets, where the cost of a given pattern set is based on system-specific nuances (e.g., exploration plans), application-specific operations (e.g., counting, computing support, etc.), and details of the data graph. 

\paragraph{Results.}
To demonstrate the generality and effectiveness of graph morphing, we integrated the pattern morphing engine with Peregrine~\cite{peregrine}~\footnote{Peregrine source code available here: \url{https://github.com/pdclab/peregrine}}, a recent state-of-the-art graph mining system. Our evaluation on different graph datasets and graph mining applications shows that pattern morphing accelerates the graph mining applications by up to 11.79$\times$, which is a significant benefit since it is on top of Peregrine's pattern-aware runtime.

\paragraph{Other Applications of Pattern Morphing.}
Graph morphing enables structure-aware algebra over patterns, and hence it is a generic technique to leverage structural relationships between different patterns. While in this paper we use graph morphing to accelerate general-purpose graph mining, graph morphing can be used for various other applications that inherently operate on graph sub-structures like incremental mining, approximate graph computations, mining on graph streams, interactive graph exploration, graph compression, and several others. The theory developed in this paper can be further enhanced with details like the application setup, execution environment, etc. to leverage it across those other applications.

\paragraph{Evolution of Pattern Morphing Project.}
With pattern-awareness deeply embedded in Peregrine, we had implemented custom pattern transformation rules designed specifically for motif counting since its inception (i.e., it has been publicly available since the first commit \texttt{d1d3df9} on April 28, 2020).
This paper generalizes those transformations as the pattern morphing technique to develop its well-defined theory and enable its wider applicability to various graph mining applications beyond just motif counting.

\section{Preliminaries}
We build on the graph mining terminology from \cite{peregrine}.

\paragraph{Graph Isomorphism.}
Given a graph $g$, $V(g)$ and $E(g)$ denote its set of vertices and edges. If $g$ is a labeled graph, $L(g)$ denotes its set of labels.

An \emph{isomorphism} between graphs $g$ and $h$ is a bijective mapping $f: V(g) \to V(h)$ that preserves edge relationships: if $(u, v) \in E(g)$, then $(f(u), f(v)) \in E(h)$. If there is an isomorphism between $G$ and $H$ we say they are \emph{isomorphic}. This is an equivalence relation.

A \emph{subgraph isomorphism} from a graph $h$ to a graph $g$ is an injective mapping $f: V(h) \to V(g)$ that preserves edge relationships: if $(u, v) \in E(h)$, then $(f(u), f(v)) \in E(g)$. We will re-define subgraph isomorphism in terms of graph patterns below.

\paragraph{Pattern Semantics.}
Graph mining applications involve finding subgraphs of interest in a given input graph (i.e., data graph). These subgraphs of interests are represented using \emph{patterns}. A \underline{pattern} is a simple connected graph, possibly with labels on its vertices. Apart from regular edges between vertices, a pattern can contain special edges called \emph{anti-edges} (originally defined in \cite{peregrine}). An \underline{anti-edge} indicates disconnections between pairs of vertices. We write $A(p)$ to denote the set of anti-edges of a pattern $p$. Note that $E(p)$ and $A(p)$ are disjoint.

We re-define subgraph isomorphism for patterns to take anti-edges into account. A subgraph isomorphism from a pattern $p$ to a graph $g$ is an injective mapping $f: V(p) \to V(g)$ that preserves edge and anti-edge relationships: if $(u, v) \in E(p)$, then $(f(u), f(v)) \in E(g)$ and if $(u, v) \in A(p)$, then $(f(u), f(v)) \not\in E(g)$. A subgraph isomorphism from a pattern into a graph is often called a \underline{match} for $p$.

We also consider subgraph isomorphisms between patterns. A subgraph isomorphism from a pattern $p$ to a pattern $q$ is an injective mapping $f: V(p) \to V(q)$ that preserves edge and anti-edge relationships. If $(u, v) \in E(p)$ then $(f(u), f(v)) \in E(q)$ and if $(u, v) \in A(p)$ then $(f(u), f(v)) \in A(q)$.

Using these definitions, anti-edges can capture how subgraphs are mapped to patterns. A \underline{vertex-induced pattern} (denoted as $p^V$) is a pattern containing anti-edges between all of its vertex-pairs that are not connected by a regular edge. Hence, vertex-induced patterns find subgraphs containing a set of vertices and all the edges incident between them. On the other hand, an \underline{edge-induced pattern} (denoted as $p^E$) is a pattern without any anti-edge. Edge-induced patterns find subgraphs containing a set of edges and their endpoints.
Note that fully connected patterns (i.e. \emph{cliques}) are simultaneously edge-induced and vertex-induced since all their vertex-pairs are adjacent.
Throughout the paper, we omit the the superscript on pattern variables when the discussion applies to both vertex-induced and edge-induced patterns.

Finally, a \underline{subpattern} of a pattern $p$ is a pattern $q$ for which there exists a subgraph isomorphism from $q$ into $p$. Conversely, a \underline{superpattern} of a pattern $p$ is a pattern $q$ such that $p$ is a subpattern of $q$.

\paragraph{Graph Mining Applications.}
There are several graph mining applications like Frequent Subgraph Mining, Motif Counting, Pattern Matching, Clique Finding, and others, that require extracting subgraphs of interest in a given data graph. These applications perform different operations over the extracted subgraphs including counting, listing (or enumerating), computing support measures, etc. A summary of these applications, along with efficient pattern programs to perform them can be found in \cite{peregrine}. Below we summarize important details for two applications that exemplify some of the differences across different applications.

\noindent
\emph{--- Motif Counting.}
This problem involves counting the occurrences of all \emph{motifs} up to a certain size  in a data graph, where motifs are connected, unlabeled graph patterns. Typically, the exploration is performed in vertex-induced manner, expressed using vertex-induced patterns. Since the aggregation operation over matches is a simple \emph{sum} operation over $n$-bit integral values, it requires constant amount of work for each match.

\noindent
\emph{--- Frequent Subgraph Mining (FSM).}
This problem involves listing all labeled patterns with $k$ edges that are frequent in a data graph. Typically, the exploration is performed in edge-induced manner, which gets expressed using edge-induced patterns. The frequency of a pattern (also called support) is often measured using anti-monotonic approaches like the \emph{minimum node image (MNI)}~\cite{fsm-mni} support measure, which ensure that if the support of a pattern is $s$, then the support of its subpatterns is at least $s$.
This property allows skipping exploration of patterns if their subpatterns are not frequent.

The MNI support measure consists of a table with a column for each group of symmetric vertices in $p$. For each $v \in V(p)$, the column corresponding to $v$ contains $m(v)$ from every match $m$. The support is defined to be the size of the smallest column in this table. Hence, to compute this support, the aggregation operation is to join tables by concatenating their respective columns. Each column contains $O(|V|)$ vertices which are merged during aggregation, so while each match incurs a constant amount of work to append to the table, combining the tables can take $O(|V|)$ work.

\begin{figure}[t]
\centering
   \begin{minipage}{0.49\linewidth}
  \subfloat[Edge-induced 4-cycles containing 4-clique, vertex-induced chordal 4-cycle, or vertex-induced 4-cycle. \label{fig-fourcyclemorphing}]{
    \includegraphics[width=1\textwidth]{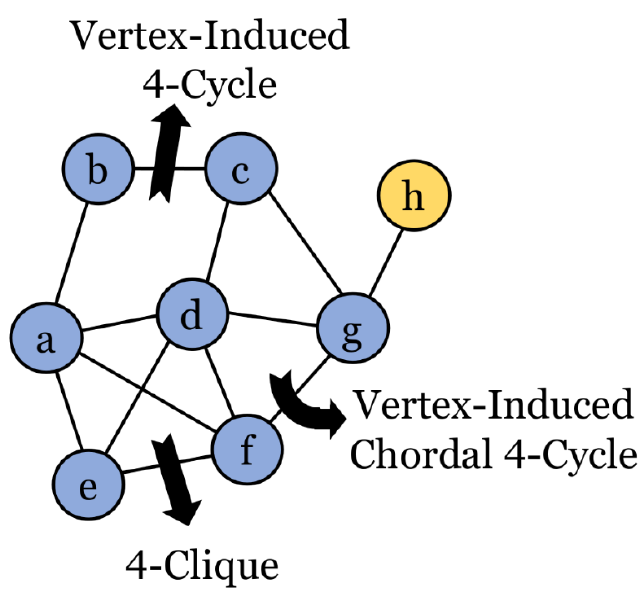}
  }
  \end{minipage}
  \hfill
  \begin{minipage}{0.5\linewidth}
  \subfloat[4-clique containing three unique edge-induced 4-cycles. \newline \label{fig-cliquetofourcyclemorphing} ]{
    \includegraphics[width=1\textwidth]{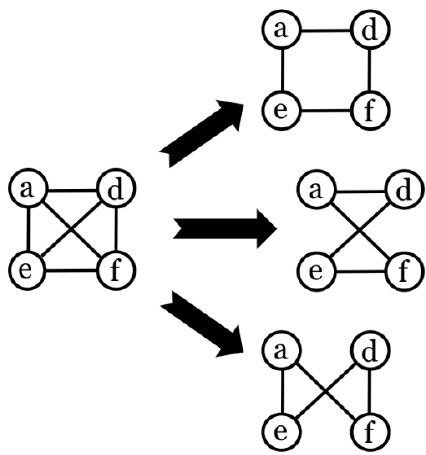}
  }
  \end{minipage}
  \caption{Identifying matches for different patterns.}
  \label{fig-unrollcollab}
\end{figure}

\section{Pattern Morphing}
\label{sec-patternmorphing}
Pattern morphing is a novel technique that enables structure-aware algebra over patterns to \emph{morph} (or convert) them into alternative patterns so that matches for the original patterns get expressed in terms of matches for other patterns. 

We first describe the main idea behind pattern morphing using illustrative examples, and then formalize the semantics of pattern morphing. 

\subsection{Intuition \& Main Idea}
Pattern morphing primarily exploits the structural similarities across different patterns. The intuition behind pattern morphing can be summarized using two key observations.

\begin{description}[leftmargin=6pt,font=\normalfont\it]
\item[---] \emph{A match for an edge-induced pattern $p^E$ on $n$ vertices is also a match for all of its subpatterns on $n$ vertices.} For example, a match for a 4-clique is also a match for an edge-induced 4-cycle, since the 4-clique contains all the edges of the 4-cycle. Note that this does not apply to vertex-induced subpatterns: even though the vertex-induced 4-cycle contains the same 4 edges, the additional two anti-edges exclude matches for 4-cliques.

\item[---] \emph{A match for a vertex-induced pattern $p^V$ is always a match for the corresponding edge-induced pattern $p^E$, since $p^V$ matches exactly the edges in $p^E$.}
\end{description}

These observations mean that we can logically partition the matches for an edge-induced pattern into disjoint sets of matches for vertex-induced patterns. For example, consider a match for an edge-induced 4-cycle. The vertices in this match may have edges between them in the data graph which are not present in the pattern. If they do not, they form a match for the vertex-induced 4-cycle (match a-b-c-d in \rfig{fig-fourcyclemorphing}). If they have exactly one extra edge, they form a match for the vertex-induced chordal 4-cycle (match d-c-g-f in \rfig{fig-fourcyclemorphing}). And finally, if they have two extra edges, they form a match for the 4-clique (match a-d-f-e in \rfig{fig-fourcyclemorphing}). Note that these situations are mutually exclusive since they depend on specific edges being present or absent.

While the above partitioning essentially allows converting the matches, a match for a given pattern can potentially contain multiple matches for another pattern. For example, a match for 4-clique contains 3 unique matches for edge-induced 4-cycle, as shown in \rfig{fig-cliquetofourcyclemorphing}. Hence, in order to convert a match for a 4-clique into a match for a 4-cycle, we must correctly map its vertices, using the subgraph isomorphisms of the 4-cycle into the 4-clique. It is crucial to note that these subgraph isomorphisms are needed to be performed only across patterns before matching starts (i.e., once only), and not across the individual matches.

\begin{figure}
  \includegraphics[width=0.49\textwidth]{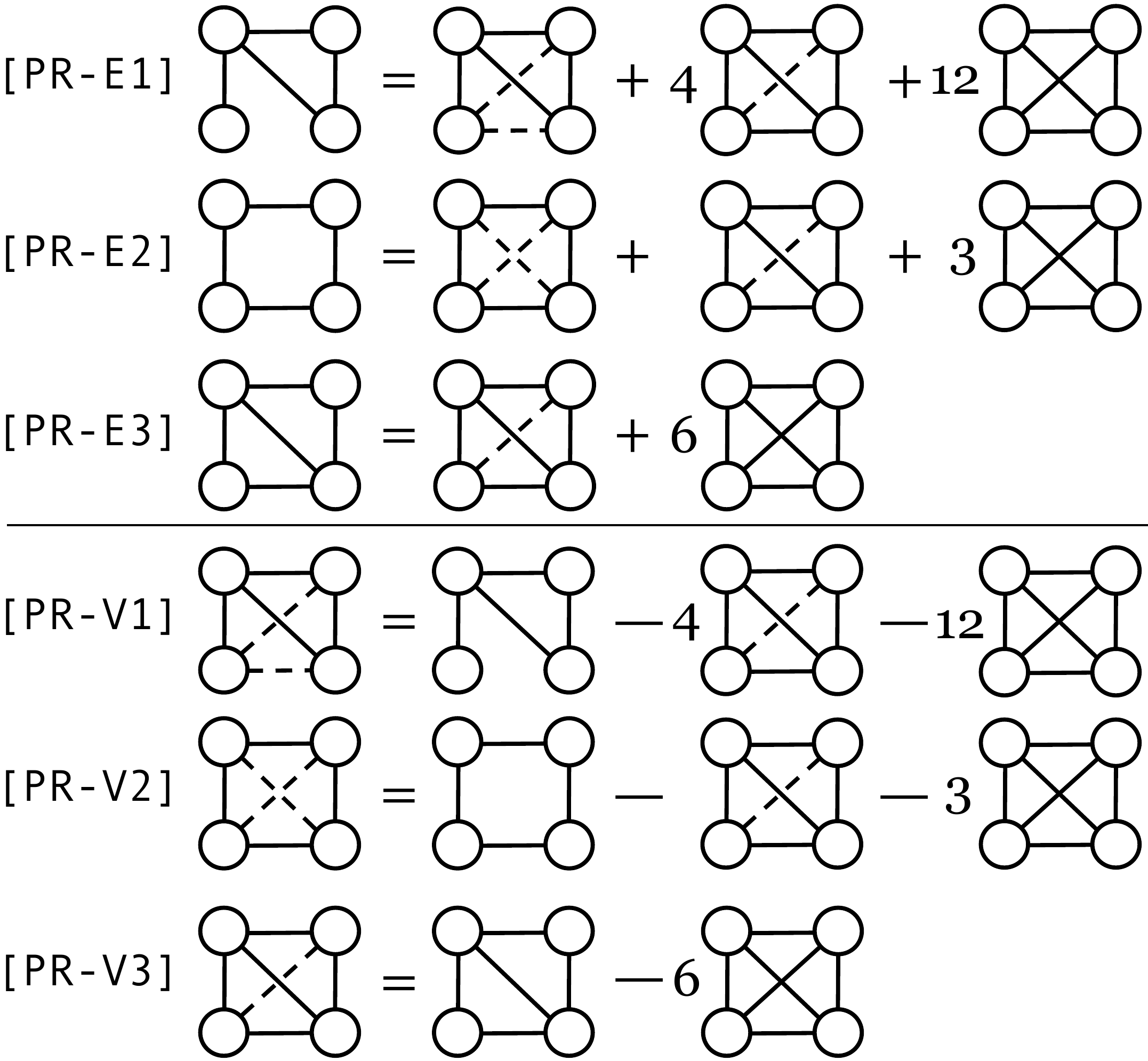}
  \caption{Sample equations resulting from pattern morphing. [PR-E1] - [PR-E3] morph edge-induced patterns (on left-hand side) to vertex-induced patterns, whereas [PR-V1] - [PR-V3] morph vertex-induced patterns to a mix of vertex-induced and edge-induced patterns. The coefficients besides the patterns indicate the number of unique matches resulting from subgraph isomorphisms. Note that patterns in the right-hand side of these equations can be substituted to generate different equations for any given pattern. The optimal equation (i.e., one that delivers best performance) for any given pattern depends on various factors including: details of data graph and pattern structures, application-specific operations, and system-specific nuances.}
  \label{fig-patternformulae}
\end{figure}

It is such kind of structure-based match conversions that pattern morphing enables.
Pattern morphing performs structure-aware algebra over patterns (and hence, their matches) to capture all the matches for a given pattern using matches of different patterns in any data graph. \rfig{fig-patternformulae} shows a few examples of how patterns can be morphed to any given pattern. The coefficients besides the patterns indicate the number of unique matches resulting from subgraph isomorphisms. Hence, for our example of morphing 4-clique to 4-cycle (discussed above), the 4-clique in Equation [PR-E2] has a coefficient 3. On the other hand, there is only a single unique subgraph isomorphism from edge-induced 4-cycle to vertex-induced 4-cycle and chordal 4-cycle, which is indicated in the same equation by no coefficients for those patterns.

Since pattern morphing mainly relies on structural similarities, the injective nature of subgraph isomorphism allows us to go the other direction as well (i.e., from edge-induced match to vertex-induced match) by simply manipulating the equations to put the desired pattern on the left-hand side. This can be seen in Equations [PR-V1] to [PR-V3] in \rfig{fig-patternformulae}.

\begin{figure}
  \includegraphics[width=0.49\textwidth]{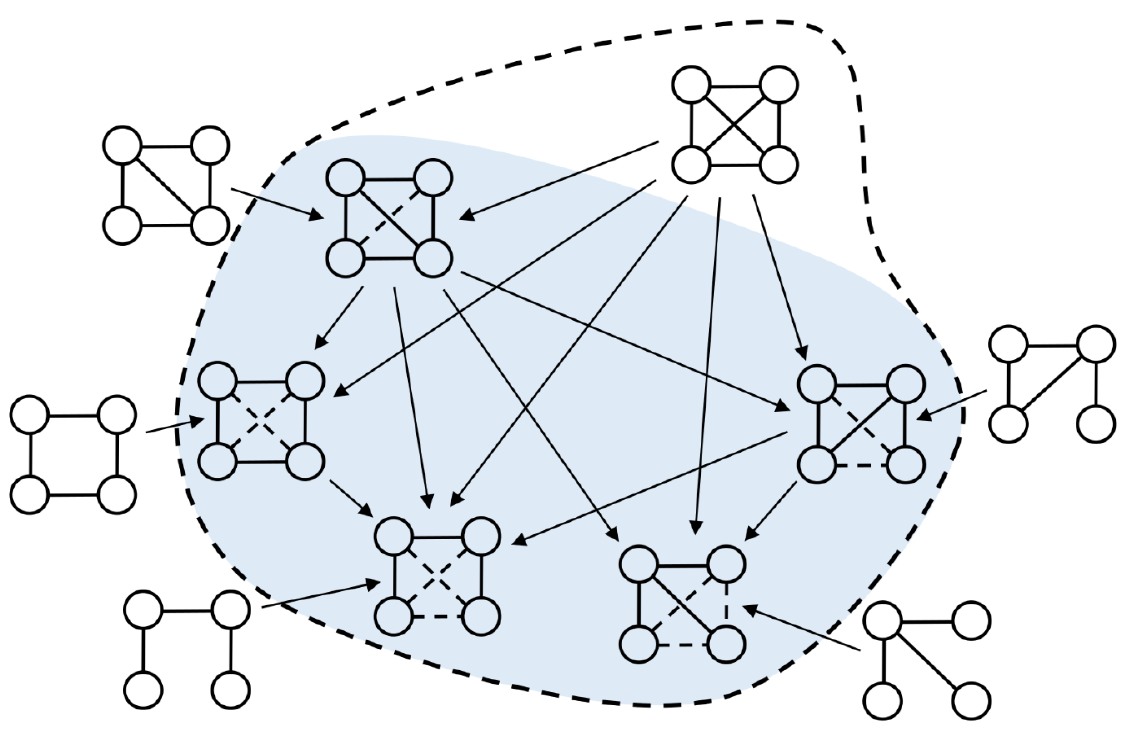}
  \caption{Motif counting example with pattern morphing. The patterns inside the dashed boundary are query patterns for motif counting. With pattern morphing, the patterns outside the shaded regions are matched, and their results are transformed to directly generate results for patterns inside the shaded region.}
  \label{fig-motifpatternplan}
\end{figure}

\paragraph{Motif Counting Example.}
The equations in \rfig{fig-patternformulae} can be directly used for 4-motif
counting: the drawing of each pattern stands for the number of matches for it,
and they get multiplied by the coefficients to capture the number of subgraph
isomorphisms across patterns. Hence, the equations can be algebraically
manipulated in order to express any set of patterns in terms of another. 
The resulting outcome of pattern morphing for 4-motif counting is shown in
\rfig{fig-motifpatternplan}: the final counts for all the 4-sized motifs can be 
directly computed using the counts for the 4-clique as well as the edge-induced 
variants of the other patterns.

\paragraph{Discussion.}
While pattern morphing can theoretically be applied across patterns of different sizes, it is often unnecessary to convert results of larger patterns to smaller patterns since it is usually expensive to match larger sized patterns instead of smaller sized patterns. Hence, we focus on morphing patterns to alternative patterns that have the same number of vertices as the patterns being morphed. 

The pattern morphing strategy discussed so far is purely based on the structural similarities across patterns. Hence, the details regarding morphing across pattern pairs can be generated once only, and reused across different data graphs and applications. In \rsec{sec-automaticpatternmorphingimpl}, we will discuss how the structure of the data graph and the complexity of the aggregation type (for different applications) can be used to select the specific morphings for a given execution.

Finally, pattern morphing can be applied to any pattern and any aggregation type; while we only showed the motif counting use case above, complex aggregation types like MNI tables~\cite{fsm-mni} (required in Frequent Subgraph Mining), or simple enumeration can also be performed with pattern morphing. As we will see next, the semantics of pattern morphing will be formalized by abstracting out details like aggregation types.

\subsection{Semantics of Pattern Morphing}
We now formalize the semantics of constructing the correct set of alternative patterns for any given (edge-induced or vertex-induced) pattern.

\subsubsection{Definitions \& Notations} 
\label{sec-definitionsnotations}
We define the key terms using which we build our analysis. 

\noindent
\emph{\underline{Subgraph Isomorphisms:}} 
Throughout this section, we will treat subgraph isomorphisms as functions between sequences of numbers, each representing a vertex. This means, for graphs $p$ and $q$ we view a subgraph isomorphism from $p$ into $q$ as a function $f : \{1,\dots,|p|\} \to \{1,\dots,|q|\}$. Hence, when $|p| = |q|$, a subgraph isomorphism is simply a permutation. For example, there are three subgraph isomorphisms from an edge-induced 4-cycle to a 4-clique, as shown in \rfig{fig-definitions}. 

Given two patterns $p$ and $q$, we define \bm{$\phi(p, q)$} to be the set of all subgraph isomorphisms from $p$ into $q$. \rfig{fig-definitions} shows two examples: $\phi(p_1^E, p_3^V)$ which contains four subgraph isomorphisms from edge-induced tailed triangle ($p_1^E$) to vertex-induced chordal 4-cycle ($p_3^V$); and, $\phi(p_2^E, p_4)$ containing the three subgraph isomorphisms from edge-induced 4-cycle ($p_2^E$) to 4-clique ($p_4$). 

\noindent
\emph{\underline{Match Sets:}}
Given a pattern $p$ (either edge-induced or vertex-induced) and a data graph $G$, $M(p, G)$ denotes the set of all matches for $p$ in data graph $G$. Each individual match in $M(p, G)$ is simply a function $m : \{1,\dots,|p|\} \to \{1,\dots,|G|\}$. Since our analysis will mainly focus on morphing patterns, which is inherently independent of any data graph, we simply use \bm{$M(p)$} as a shorthand notation (i.e., we drop $G$). 

\noindent
\emph{\underline{Permuting Matches based on Subgraph Isomorphisms:}}
Given a data graph $G$ and two patterns $p$ and $q$ with $n$ vertices, we write \bm{$M(q) \circ \phi(p, q)$} to mean the permutations of the vertices in each match $m \in M(q)$ according to the subgraph isomorphisms from $p$ into $q$. Note that the permuted matches are nothing but matches of $p$ in $G$. Specifically, since a match for $q$ in $G$ is an injective function $m : \{1,\dots,n\} \to \{1,\dots,|G|\}$, and $\phi(p, q)$ consists of bijective functions $f : \{1,\dots,n\} \to \{1,\dots,n\}$:
  \[M(q) \circ \phi(p, q) = \{m \circ f : m \in M(q) \land f \in \phi(p, q)\}\]
where $m \circ f$ means permuting the domain of $m$ according to $f$.

\begin{figure}[t]
\includegraphics[width=0.49\textwidth]{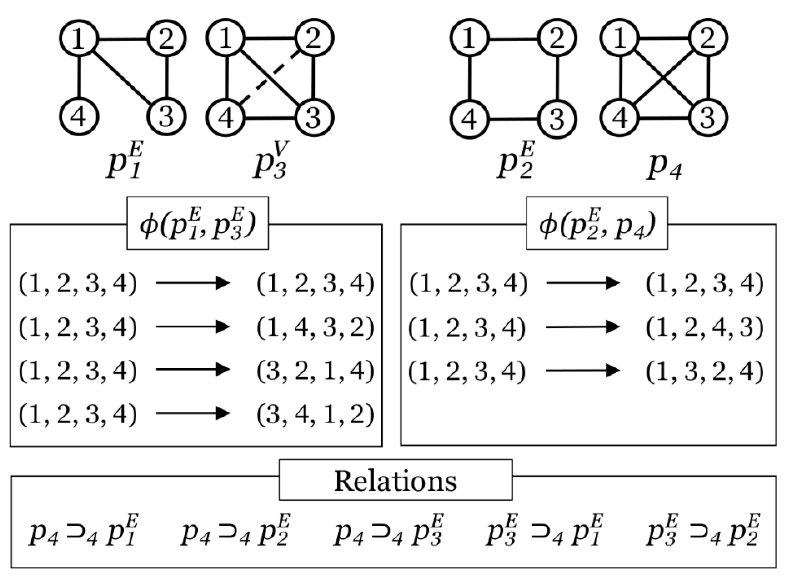}
\caption{Examples of subgraph isomorphisms across patterns, and non-isomorphic superpatterns.}
\label{fig-definitions}
\end{figure}

For example, consider the matches for 4-clique (pattern $p_4$ in \rfig{fig-definitions}) in the data graph from \rfig{fig-fourcyclemorphing}. Here, $M(p_4) = \{\{1\to a, 2 \to d, 3 \to f, 4 \to e\}\}$ and hence:
\begin{align*}
M(p_4) \circ \phi(p_2^E, p_4) = 
\begin{cases}
\begin{rcases}
\{1\to a, 2 \to d, 3 \to f, 4 \to e\}, \\ 
\{1\to a, 2 \to d, 4 \to f, 3 \to e\}, \\ 
\{1\to a, 3 \to d, 2 \to f, 4 \to e\}
\end{rcases}
\end{cases}
\end{align*}
Each of $m \circ f \in M(p_4) \circ \phi(p_2^E, p_4)$ is a match for $p_2^E$ (i.e., edge-induced 4-cycle) in subgraph of $G$ containing the four vertices $a$, $d$, $f$ and $e$.

\noindent
\emph{\underline{Non-Isomorphic Superpatterns:}}
Finally, given patterns $p$ and $q$, we write \bm{$q \supset_n p$} if $q$ is a non-isomorphic superpattern of $p$ containing at most $n$ vertices. In \rfig{fig-definitions}, 4-clique is a non-isomorphic superpattern of the remaining three patterns, while vertex-induced chordal 4-cycle is a non-isomorphic superpattern of edge-induced 4-cycle and edge-induced tailed triangle. 

\subsubsection{Transforming Matches for Pattern Morphing}
The key idea behind pattern morphing is to transform the matches of a given pattern into matches for other patterns. This is captured by the \emph{Match Conversion Theorem}. 

\begin{thm}[\textbf{Match Conversion Theorem}]
\label{thm-conversion}
  Let $p^E$ be an edge-induced pattern with $n$ vertices, and $p^V$ be its vertex-induced variant. Then,
  \[M(p^E) \ = \ M(p^V) \ \ \cup \bigcup_{q^E \supset_n p^E} M(q^V)\circ \phi(p^E, q^E)\]
\end{thm}

\begin{proof}
  Since we are proving set equality, we will first show how every match in $M(p^E)$ is
  contained in the set on the right-hand side of the equation, and then show that 
  $M(p^E)$ contains every match from the right-hand side.

  Let $m$ be a match in $M(p^E)$, and $q^V$ be the pattern of the subgraph induced 
  by the image of $m$. $q^V$ must have at least as many edges as $p^E$, since it was
  induced by a match for $p^E$. 
  
  \noindent \emph{-- Case 1: \ } 
  If $q^V$ has the same number of edges as $p^E$, 
  $q^V$ is isomorphic to $p^V$ and hence $m \in M(p^V)$.
  
  \noindent \emph{-- Case 2: \ } 
  Otherwise, $q^V$ has more edges than $p^E$. Consider the edge-induced pattern
  $q^E$ corresponding to $q^V$. $q^E$ must be a superpattern of $p^E$, since
  $q^V$ has more edges than $p^E$ but contains all the edges of $p^E$. This
  means $\phi(p^E, q^E)$ is non-empty. For any $f \in \phi(p^E, q^E)$, observe
  that $m\circ f^{-1} : V(q^E) \to V(p^E) \to G$ is a match for $q^V$, since
  $V(q^E) = V(q^V)$ and $q^E$ was obtained from the induced pattern $q^V$. 
   This means $m\circ f^{-1} \in M(q^V)$, and hence 
   $m \circ f^{-1} \circ f = m \in M(q^V)\circ \phi(p^E, q^E)$. 

   Therefore, we showed $M(p^E)$ is a subset of the right-hand side of the equation.
   Next, we will prove that the right-hand side is contained within $M(p^E)$.
   
  First, note that a match for $p^V$ is trivially a match for $p^E$, since $p^E$
  and $p^V$ differ only in anti-edges, so $M(p^V) \subseteq M(p^E)$.

  Next, take any match $m \in M(q^V)$ where $q^E \supset_n p^E$. $m$ is also in
  $M(q^E)$ by the same reasoning as above. But then for any $f \in \phi(p^E, q^E)$,
  $m \circ f : V(p^E) \to V(q^E) \to G$ is a match for $p^E$ since any match for $q^E$ 
  contains all edges required for matching $p^E$. 
  Hence, $M(q^V) \circ \phi(p^E, q^E) \subseteq M(p^E)$.

  Taking the union of $M(p^V)$ and $M(q^V) \circ \phi(p^E, q^E)$ for each
  $q^E \supset_n p^E$ gives the set of matches that are contained in $M(p^E)$.
\end{proof}

The Match Conversion Theorem reflects a useful relationship between edge-induced and
vertex-induced patterns. In fact, although the theorem is stated in terms of morphing from vertex-induced to edge-induced, we can move in the other direction as well:

\begin{cor}
  \label{cor-vinduced}
  Let $p^E$ be an edge-induced pattern with $n$ vertices, and $p^V$ be its vertex-induced variant. Then,
  \[M(p^V) = M(p^E) \setminus \bigcup_{q^E \supset_n p^E} M(q^V)\circ \phi(p^E, q^E)\]
\end{cor}

\begin{proof}
  Let the set $\bigcup_{q^E \supset_n p^E} M(q^V)\circ \phi(p^E, q^E)$ be denoted as $B$.
  From \rthm{thm-conversion}, we have $M(p^E) = M(p^V) \cup B$. Notice that if $M(p^V)$
  is disjoint from all $M(q^V)\circ \phi(p^E, q^E)$ where $q^E \supset_n p^E$, then we 
  could take the set difference on both sides of the equation from \rthm{thm-conversion} 
  with $B$ to prove the corollary, simply because no element of $M(p^V)$ would be 
  removed by the set difference operation.

  It remains to prove that $M(p^V)$ is disjoint from the various $M(q^V) \circ
  \phi(p^E, q^E)$ where $q^E \supset_n p^E$. We prove this by contradiction. 

  Let $m$ be a match in $M(p^V) \cap M(q^V) \circ \phi(p^E, q^E)$ for any $p^V$ and 
  some $q^V$ where $q^E \supset_n p^E$.   
  First note that if $p^V$ is a clique, there is no $q^E \supset_n p^E$, and $m$
  cannot exist, so we are done. Instead, suppose $p^V$ is not a clique, and thus
  has at least one anti-edge.
  Since $q^E \supset_n p^E$, for each $f \in \phi(p^E, q^E)$ there is an edge
  $(f(u), f(v))\in E(q^E)$ such that $(u, v)\not\in E(p^E)$. This means $(u, v)$
  forms an anti-edge constraint in $p^V$, and $(f(u), f(v))$ forms an edge
  constraint in $q^V$.
  As $m \in M(q^V)\circ \phi(p^E, q^E)$, it can be written in the form $m = m'
  \circ f$ where $m' \in M(q^V)$ and $f \in \phi(p^E, q^E)$. But then $((m'\circ
  f)(u), (m'\circ f)(v))$ must be an edge in $G$, since $(f(u), f(v)) \in
  E(q^V)$ and $m'$ is a match for $q^V$. 
  This directly contradicts the anti-edge constraint in $p^V$ that we established 
  above. 

  Hence $M(p^V) \cap M(q^V)\circ\phi(p^E, q^E) = \emptyset$ for all $q^E
  \supset_n p^E$, as desired.
\end{proof}

By recursively substituting in the equation in \rcor{cor-vinduced}, we
can express the matches of a vertex-induced pattern in terms of the matches for
edge-induced patterns. Since the final superpattern of any non-clique pattern is
a clique, this recursive substitution terminates.

\paragraph{Discussion.}
Using \rthm{thm-conversion}, we can easily materialize edge-induced matches from vertex-induced matches. Similarly, we can convert edge-induced matches to vertex-induced ones using \rcor{cor-vinduced}; however, materializing these vertex-induced matches would require first materializing $M(p^E)$ for computing the set difference, which may not always be feasible. Nevertheless, the ability to accurately convert edge-induced matches to vertex-induced ones (and vice-versa) enables directly computing results for any given pattern from its alternative pattern sets, as shown next.

\subsubsection{Aggregation with Pattern Morphing}
Graph mining applications often compute aggregated statistics (e.g., counts, support, etc.) based on the matches that get explored from the data graph. We capture how Pattern Morphing naturally enables efficient computation of such aggregated statistics. 

Let $a = (\lambda, \oplus)$ be the aggregation in graph mining applications where $\lambda$ is a map from matches to an aggregation value and $\oplus$ is a commutative operator for combining aggregation values. For a set of matches $M(p)$, we write $a(M(p))$ as a shorthand for $\bigoplus\limits_{m \in M(p)} \lambda(m)$. 

Note that our definition of aggregation captures functions whose images can be summed. The simplest example is counting, where $\lambda(M) = |M|$ for a set of matches $M$ and the $\oplus$ operator is the traditional integer sum. For more complicated applications like FSM, $\lambda$ computes the MNI table of a set of matches and the $\oplus$ operator joins tables on columns. In addition, we define a \emph{permute} operator $\circ_{*}$ for aggregation values to account for the permutations according to $\phi(q, p)$ (similar to $\circ$ defined on matches in \rsec{sec-definitionsnotations}). In the counting example, for a match $m$ and a permutation $f$, $\circ_*$ can be defined as $a(m) \circ_{*} f = a(m)$. In FSM, on the other hand, $a(m) \circ_{*} f$ permutes the columns of the MNI table of a match $m$ according to permutation $f$.

Using this abstraction, we prove the following critical theorem that enables directly computing aggregation values using matches for morphed patterns.

\begin{thm}[\textbf{Aggregation Conversion Theorem}]
  \label{thm-agg}
  Let $a = (\lambda, \oplus)$ be the aggregation in graph mining applications. For any match $m$, permutation $f$ on $n$ vertices, and permute operator $\circ_{*}$, if \ $a(m\circ f) = a(m) \circ_{*} f$, then:
  \[a(M(p^E)) = a(M(p^V)) \oplus \left( \bigoplus_{q^E \supset_n p^E} \ \bigoplus_{f \in \phi(p^E, q^E)} \hspace{-0.1in} a(M(q^V))\circ_{*} f \right) \]
\end{thm}

\begin{proof}
  The equation follows immediately from \rthm{thm-conversion}.
  \begin{align*}
    a(M(p^E)) \ &= \ a\left(M(p^V) \ \ \cup \bigcup_{q^E \supset_n p^E} M(q^V)\circ \phi(p^E, q^E)\right)\\
                &= \ a(M(p^V)) \oplus \left(\bigoplus_{q^E \supset_n p^E} a(M(q^V)\circ \phi(p^E, q^E)) \right)\\
                &= \ a(M(p^V)) \oplus \left(\bigoplus_{q^E \supset_n p^E} \ \bigoplus_{f \in \phi(p^E, q^E)} \hspace{-0.1in} a(M(q^V)) \circ_* f \right)
  \end{align*}
\end{proof}

\rthm{thm-agg} demonstrates the benefit of pattern morphing as a significant amount of matches no longer need to be materialized to compute the aggregation values. In other words, we can obtain the results of an aggregation over a pattern by using the aggregations over a completely different set of patterns.
The performance gains achieved due to this can be captured as follows.

\begin{cor}
  \label{cor-time}
  Let $T(p)$ be the time it takes to find all matches of a pattern $p$ with
  $n$ vertices and apply aggregation  $a = (\lambda, \oplus)$ that satisfies \rthm{thm-agg}.
  Then for patterns $p^E$ and $q^E \supset_n p^E$, 
  the aggregation \ $a(M(q^V) \circ \phi(p^E, q^E))$ can be computed in time $O(T(q^V)) + O(|\phi(p^E, q^E)|)$.
\end{cor}
\begin{proof} 
  By definition, we already know:
  \begin{align*}
    a(M(q^V)\circ\phi(p^E,q^E)) &= \bigoplus_{f \in \phi(p^E,q^E)} \hspace{-0.1in} a(M(q^V)) \circ_* f
  \end{align*}

  This means that we can compute $a(M(q^V)\circ \phi(p^E,q^E))$ by first
  computing $a(M(q^V))$ and then applying the permutations $f \in \phi(p^E,
  q^E)$. Hence, $O(T(q^V))$ time spent computing $g(M(q^V))$ and $O(|\phi(p^E,
  q^E)|)$ time spent adjusting the results.
\end{proof}

\rcor{cor-time} shows how pattern morphing boosts performance by reducing the number of matches we need to produce for each $q^E \supset_n p^E$ by a factor of $|\phi(p^E, q^E)|$. 

Composing \rthm{thm-agg} and \rcor{cor-vinduced} leads to the expected result that aggregations of a vertex-induced pattern can be computed using the aggregations of edge-induced patterns, with the added restriction that the aggregation function's image must also be additive. We skip the proof since it is nearly identical to that of \rcor{cor-time}.
Additionally, the alternative pattern set can contain a combination of
vertex-induced and edge-induced patterns by converting the intermediate
aggregations through recursive applications of the theorem.

\section{Implementation \& Evaluation}

\subsection{Automatic Pattern Morphing Implementation}
\label{sec-automaticpatternmorphingimpl}
The pattern morphing theory developed in \rsec{sec-patternmorphing} can be applied in practice to accelerate graph mining tasks in any general-purpose graph mining system. Since pattern morphing guides conversion of results across different patterns, it can be added as an `external module' to any graph mining system (i.e., without changing the implementation of the system's subgraph exploration strategies). Such a module would: (a) generate alternative pattern sets for the input pattern sets; and, (b) transform the aggregation/operation results after matches get explored. 

For this evaluation, we developed an automatic pattern morphing engine for Peregrine, which is the state-of-the-art pattern-aware graph mining system. Peregrine, being pattern-aware, already exploits efficient matching algorithms to generate optimized matching plans for different patterns. Adding pattern morphing as an external module enables leveraging the optimized plans across different patterns. 

There can be multiple alternative pattern sets for any given set of patterns, and those different alternative pattern sets can deliver different performance depending on the specific patterns they contain. Hence, selecting the right alternative pattern sets for a given set of patterns is crucial to deliver high performance. This can be addressed by developing a `pattern morphing optimizer' that minimizes the cost of pattern sets to construct the best alternative pattern sets. Here, the cost of a given pattern set must capture three main factors:
\begin{enumerate}[leftmargin=0.15in]
\item The nuances of exploration strategies used by the underlying system. For instance, on Peregrine this would include the work done in adjacency list intersection and difference operations, positive effects of symmetry breaking, and the custom pattern-specific optimizations it employs. On other systems that may not be fully pattern-aware (e.g., due to lack of symmetry breaking), the cost of exploring patterns would be different.
\item The details of performing application-specific operations on matches. For instance, counting patterns takes $\BigO{1}$ for groups of matches, whereas computing MNI support for FSM takes $\BigO{|V(G)|}$ time. With pattern morphing, the conversion costs for these operations must also be considered. 
\item The details of the data graph. This would not only include the structural details of the data graph like degree distribution and connectivity information, but also include application-specific properties in the data graph like label distributions (required for FSM).
\end{enumerate}
We evaluate the effectiveness of such a cost-model based pattern morphing optimizer by incorporating it with our pattern morphing engine for Peregrine.

\subsection{Experimental Setup}
\paragraph{System.}
We evaluate the pattern morphing engine in Peregrine on a Google Cloud \texttt{n2-highcpu-32} instance which consists of a 2.8GHz Intel Cascade Lake processor with 32 logical cores and 32GB of RAM.

\paragraph{Datasets.}
\rtab{table-datasets} lists the data graphs used in our evaluation. Mico (MI) is a co-authorship graph labeled with each author's research field. Patents (PA) is a patent citation graph, where each patent is labeled with the year it was granted. Youtube (YT) consists of videos crawled by~\cite{youtube} from 2007--2008, with edges between related videos. Videos are labeled according to their ratings, like in~\cite{fractal}. Orkut (OK) is an unlabeled social network graph where edges represent friendships between users. All these datasets have been used to evaluate previous systems~\cite{arabesque,fractal,automine,peregrine}.

\begin{figure}
  \label{fig-evaluation-patterns}
  \includegraphics[width=0.49\textwidth]{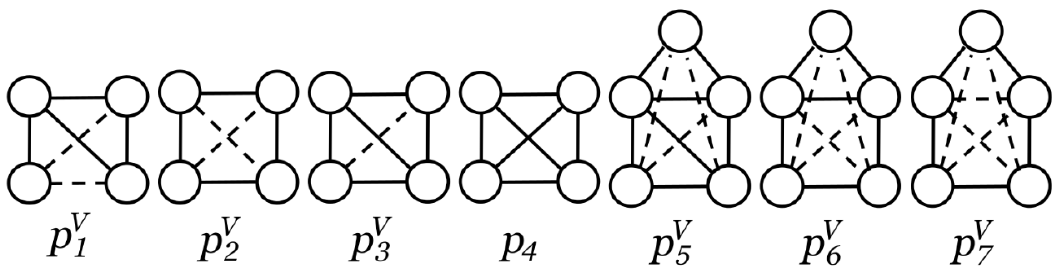}
  \caption{Patterns used in evaluation.}
  \label{fig-evaluation-patterns}
\end{figure}

\paragraph{Applications.}
We run experiments on a wide array of applications: counting motifs with 3 and 4 vertices, size 3 FSM, and matching vertex-induced versions of the patterns in \rfig{fig-evaluation-patterns}. Throughout the evaluation, we use the following notations for different variants:
\begin{itemize}[leftmargin=0.15in]
\item \textbf{\nomorph{}:} this is Peregrine without pattern morphing, where the query patterns for each of the applications are directly used to explore the data graph.
\item \textbf{\naivemorph{}:} this is Peregrine with pattern morphing, where the query patterns are morphed so that edge-induced input patterns are morphed into vertex-induced patterns, and vertex-induced input patterns are morphed into edge-induced patterns.
\item \textbf{\smartmorph{}:} this is Peregrine with the pattern-morphing optimizer that constructs the best alternative pattern sets by minimizing their costs. 
\end{itemize}

\begin{table}[t]
\small
  \begin{tabular}{l r r c r r}
    \multicolumn{1}{c}{\multirow{2}{*}{$G$}} & \multirow{2}{*}{$|V(G)|$} & \multirow{2}{*}{$|E(G)|$} & \multirow{2}{*}{$|L(G)|$} & Max. & Avg.  \\
     &  &  &  & Deg. & Deg. \\
    \midrule
(MI) Mico~\cite{mico}       & 100K    & 1M     & 29  & 1359  & 22   \\ 
(PA) Patents~\cite{patents} & 3.7M    & 16M    & 37  & 789   & 10   \\ 
(YT) YouTube~\cite{youtube} & 6.9M    & 44M    & 38  & 4039  & 12   \\
(OK) Orkut~\cite{snap}      & 3M      & 117M   & --- & 33133 & 76   \\
    \midrule
  \end{tabular}
  \caption{Real-world graphs used in evaluation. \\ '---' indicates unlabeled graph.}
  \label{table-datasets}
\end{table}

\begin{table}
\small
  \begin{tabular}{c c | r | r | r}
    App & $G$ &
    \nomorph{} & 
    \naivemorph{} &
    \smartmorph{}  \\
    \midrule
          3-MC     & MI    & 0.14             & \textbf{0.07}      & \textbf{0.07} ($2\times$)    \\
                   & PA    & 0.97             & \textbf{0.52}      & \textbf{0.52} ($1.87\times$)  \\
                   & YT    & 6.83             & \textbf{2.06}      & \textbf{2.06} ($3.32\times$)  \\
                   & OK    & 30.15            & \textbf{10.57}     & \textbf{10.57} ($2.85\times$) \\
            \midrule
          4-MC     & MI    & 16.53            & \textbf{3.30}      & \textbf{3.30} ($5.01\times$) \\
                   & PA    & 21.27            & \textbf{6.10}      & \textbf{6.10} ($3.49\times$) \\
                   & YT    & 178.04           & \textbf{41.78}     & \textbf{41.78} ($4.26\times$) \\
                   & OK    & 37017.93         & \textbf{3138.76}   & \textbf{3138.76} ($11.79\times$) \\
            \midrule
          $p_1^V$  & MI    & 4.44             & \textbf{1.16}      & \textbf{1.16} ($3.82\times$) \\
                   & PA    & \textbf{1.11}    & 1.17               & \textbf{1.11}  \\
                   & YT    & 10.70            & \textbf{9.27}      & \textbf{9.27} ($1.15\times$) \\
                   & OK    & 913.70           & \textbf{177.79}    & \textbf{177.79} ($5.14\times$) \\
            \midrule
          $p_2^V$  & MI    & \textbf{1.89}    & 3.15               & \textbf{1.89}  \\
                   & PA    & \textbf{3.97}    & 4.35               & \textbf{3.97}  \\
                   & YT    & \textbf{31.91}   & 34.18              & \textbf{31.91} \\
                   & OK    & \textbf{1730.37} & 3010.42            & \textbf{1730.37} \\
            \midrule
          $p_3^V$  & MI    & 3.04             & \textbf{1.08}      & \textbf{1.08} ($2.81\times$) \\
                   & PA    & \textbf{0.82}    & 0.95               & \textbf{0.82} \\
                   & YT    & 8.07             & \textbf{6.79}      & \textbf{6.79} ($1.16\times$) \\
                   & OK    & 403.41           & \textbf{156.99}    & \textbf{156.99} ($2.57\times$) \\
            \midrule
          $p_5^V$  & MI    & 84.08            & \textbf{65.62}     & \textbf{65.62} ($1.28\times$) \\
                   & PA    & \textbf{8.32}    & 10.77              & \textbf{8.32}  \\
                   & YT    & \textbf{42.50}   & 83.03              & \textbf{42.50} \\
                   & OK    & 41130.03         & \textbf{9489.79}   & \textbf{9489.79} ($4.33\times$) \\
            \midrule
          $p_6^V$  & MI    & \textbf{31.98}   & 88.08              & \textbf{31.98} \\
                   & PA    & 35.23            & \textbf{25.34}     & \textbf{25.34} ($1.39\times$) \\
                   & YT    & \textbf{109.48}  & 182.31             & \textbf{109.48} \\
                   & OK    & ---              & \textbf{27014.10}  & \textbf{27014.10}         \\
            \midrule
          $p_7^V$  & MI    & \textbf{23.56}   & 423.87             & \textbf{23.56} \\
                   & PA    & \textbf{14.95}   & 58.37              & \textbf{14.95} \\
                   & YT    & \textbf{67.03}   & 429.97             & \textbf{67.03} \\
            \midrule
           $p_2^E$ & MI    & 2.09             & 4.75               & \textbf{1.93} ($1.08\times$) \\
                   & PA    & \textbf{3.29}    & 4.16               & \textbf{3.29} \\
                   & YT    & 27.84            & 21.22              & \textbf{16.69} ($1.67\times$) \\
                   & OK    & 2896.95          & 2331.59            & \textbf{1841.42} ($1.57\times$) \\
            \midrule
$\{p_2^E, p_3^E\}$ & MI    & 3.92             & 2.70               & \textbf{1.93} ($2.03\times$) \\
                   & PA    & \textbf{3.29}    & 3.62               & \textbf{3.29} \\
                   & YT    & 18.82            & 17.95              & \textbf{16.69} ($1.28\times$) \\
                   & OK    & 2908.70          & 1997.15            & \textbf{1845.09} ($1.58\times$) \\
            \midrule
$\{p_5^V, p_6^V\}$ & MI    & 116.06           & \textbf{88.08}     & \textbf{88.08} ($1.32\times$) \\
                   & PA    & 43.55            & \textbf{25.34}     & \textbf{25.34} ($1.56\times$) \\
                   & YT    & \textbf{151.98}  & 182.31             & \textbf{151.98} \\
                   & OK    & ---              & \textbf{27014.10}  & \textbf{27014.10} \\
            \midrule
          3-FSM    & MI    & 127.16           & 101.60             & \textbf{71.31} ($1.78\times$) \\
                   & PA    & 644.74           & 3772.41            & \textbf{639.75} ($1.01\times$) \\
                   & YT    & \textbf{491.07}  & 550.29             & \textbf{491.07}  \\
    \midrule
\end{tabular}
\caption{Execution times in seconds (including time spent morphing patterns) with and without pattern morphing. 
  \\`---' represents executions that did not finish within 24 hours.}

\label{table-results}
\end{table}

\subsection{Overview of Performance Results}
\rtab{table-results} shows the performance of pattern morphing for all applications across different data graphs. Overall, \naively{} using pattern morphing shows a benefit of 1.16-11.79$\times$, with potential slowdowns in only few cases.  
The cost-based pattern morphing correctly captures those cases and retains high performance benefits across all cases. Below we analyze how pattern morphing improves each of the applications.

\subsection{Pattern Morphing for Motif Counting} 
The greatest benefits for pattern morphing occur in the motif counting application, since all superpatterns are already included in the input pattern set. As a result, the morphing engine always generates an optimal morphing set. Furthermore, counting is a very cheap aggregation, which makes morphing vertex-induced patterns very profitable in terms of the set operation versus aggregation tradeoff.

Pattern morphing yielded a better execution time in all the motif counting experiments. For 4-motif counting on the large Orkut graph, pattern morphing yields $11.79\times$ speedup over the baseline that doesn't use pattern morphing.

\subsection{Pattern Morphing for Pattern Matching}
Matching individual patterns is a test of the worst case for pattern morphing, since several extra superpatterns need to be matched in order to morph a given single pattern. Despite this, our cost-based patten morphing optimizer finds opportunities for a $5\times$ speedup when matching $p_1^V$ on Orkut, and $4.33\times$ when matching $p_5^V$ on Orkut. Additionally, with pattern morphing, Peregrine was able to match $p_6^V$ on Orkut in 7.5 hours, whereas without morphing the execution did not finish within 24 hours.

We also match groups of patterns to observe whether it amortizes the cost of the extra superpatterns. Interestingly, when matching only $p_6^V$, morphing it would lead to a slowdown on the MiCo graph, but when the input pattern set consists of both $p_5^V$ and $p_6^V$, pattern morphing yields a $1.32\times$ speedup.

\begin{table}
  \label{table-altset}
  \begin{tabular}{c c | l}
    App & $G$ & Alt. Set\\
    \midrule
  $p_1^V$  & MI & $\{p_1^E, p_3^E, p_4\}$ \\
           & PA & $\{p_1^E, p_3^E, p_4\}$ \\
           & YT & $\{p_1^E, p_3^E, p_4\}$ \\
           & OK & $\{p_1^E, p_3^E, p_4\}$ \\
    \midrule
  $p_2^V$  & MI & $\{p_2^V\}$ \\
           & PA & $\{p_2^V\}$ \\
           & YT & $\{p_2^V\}$ \\
           & OK & $\{p_2^V\}$ \\
    \midrule
  $p_2^E$  & MI & $\{p_2^V, p_3^E, p_4\}$ \\
           & PA & $\{p_2^E\}$ \\
           & YT & $\{p_2^V, p_3^E, p_4\}$ \\
           & OK & $\{p_2^V, p_3^E, p_4\}$ \\
    \midrule
  $p_3^V$  & MI & $\{p_3^E, p_4\}$ \\
           & PA & $\{p_3^V\}$ \\
           & YT & $\{p_3^E, p_4\}$ \\
           & OK & $\{p_3^E, p_4\}$ \\
    \midrule
$\{p_2^E, p_3^E\}$ & MI & $\{p_2^V, p_3^E, p_4\}$ \\
                   & PA & $\{p_2^E, p_3^E\}$      \\
                   & YT & $\{p_2^V, p_3^E, p_4\}$ \\
                   & OK & $\{p_2^V, p_3^E, p_4\}$ \\
    \midrule
  \end{tabular}
  \caption{Alternative pattern sets used in \smartmorph{}.}
  \label{tab-altsets}
\end{table}

\rtab{tab-altsets} lists the alternate pattern sets selected by \smartmorph{} for some of the input patterns across different data graphs. As we can see, the cost-based pattern morphing optimizer selects different outcomes that minimize the pattern set costs. For instance, it correctly detects that morphing $p_3^V$ is beneficial for all graphs except the Patents graph. Similarly, while \naivemorph{} ends up morphing both patterns in $\{p_2^E, p_3^E\}$, \smartmorph{} correctly morphs only $p_2^E$ for all graphs except the Patents graph.

\subsection{Pattern Morphing for Frequent Subgraph Mining}
We ran 3-FSM on the labeled datasets MiCo, Patents, and YouTube with support measures of 4000, 23000, and 300000, respectively (similar to \cite{arabesque,fractal,peregrine}). On the MiCo graph, pattern morphing reduces the execution time by 43\%, but the benefit is marginal on Patents, and none on YouTube since the cost-based pattern morphing optimizer ends up choosing not to morph the input pattern set. These performance differences are inherent to the data graph, and depend on the distribution of labels among vertices.

\section{Related Work}
There has been a variety of research to develop efficient graph mining solutions. To the best of our knowledge, this is the first formal treatment of generic structure-based transformations of input patterns in graph mining systems.

\paragraph{General-Purpose Graph Mining.} 
Several general-purpose graph mining systems have recently emerged in the literature~\cite{arabesque,fractal,rstream,peregrine,automine,pangolin}.
Arabesque~\cite{arabesque} and Fractal~\cite{fractal} are exploration-based
general-purpose graph mining systems which mine subgraphs through iterative
extensions by edges or vertices. RStream~\cite{rstream} is an out-of-core graph mining
system which structures computations as relational operations on tables of
subgraphs and stores intermediate state on disk. Pangolin~\cite{pangolin} is a
recent exploration-based graph mining system which supports offloading
computation to GPU. ASAP~\cite{asap} is a general-purpose approximate
graph mining system allowing users to navigate the tradeoff between error
and performance by analyzing graph mining computations.
\cite{automine} compiles input patterns into exploration programs for graph mining by applying tournament algorithm to select efficient set operation schedules. 
Peregrine~\cite{peregrine} is a pattern-aware general-purpose
graph mining system which exploits structural and label properties of input
patterns to efficiently mine their matches.

All these works focus on efficiently processing graph mining applications
as expressed by the user, by optimizing the exploration plans for the 
subgraph structures of interest. In contrast, we focus on how graph mining 
applications can be transformed by morphing the subgraph structures of interest
to enable more efficient execution. 
Pattern morphing can be integrated with these systems as an add-on module 
to improve their performance.

\paragraph{Motif Counting.} 
A myriad of research has been conducted on
efficient algorithms for mining motifs. \cite{po} developed a technique to break
symmetries of patterns by enforcing a partial-ordering on pattern vertices.
\cite{efficient-graphlets} is a fast, parallel algorithm for counting size 3 and
4 motifs using combinatorial identities. RAGE~\cite{rage} provides a method for
efficiently computing edge-induced size-4 motifs, and equations for converting
the results to those for vertex-induced motifs. \cite{orbit-counting}
efficiently lists the automorphism groups of pattern vertices and uses them to
compute exact counts for motifs with 2-5 vertices. There is also research
on techniques for computing approximate counts of motifs~\cite{graft,approxg}.

All of these works focus on counting patterns of fixed size using 
explicit formulae for those specific patterns. 
Our work differs by developing a general algebra for morphing arbitrary
patterns and transforming arbitrary graph mining applications. 
Pattern
morphing is not restricted by pattern types, pattern size, or application semantics.

\paragraph{Pattern Matching.} 
These works aim to devise more efficient
subgraph isomorphism solutions using sophisticated analysis of data and query
graphs~\cite{dualsim,turboiso,turboflux,daf,prunejuice,qfrag,postponecart,ceci}.
DualSIM~\cite{dualsim} is an out-of-core pattern matching system that uses
pattern-core decomposition to efficiently map pattern vertices onto data
vertices. CECI~\cite{ceci}, a distributed pattern matching system, builds
clustered indexes in the data graph based on the query pattern topology and uses
these indexes to quickly materialize matches. \cite{turboiso,daf,postponecart}
decompose the query graph into tree and non-tree edges and match them
efficiently guided by sophisticated heuristics. These works are focused on
algorithms for pattern matching, which is orthogonal to the issue of what
patterns to match.

\paragraph{Frequent Subgraph Mining.} 
GraMi~\cite{mico} performs
efficient FSM by enumerating a minimal set of matches to determine frequency of
patterns. ScaleMine~\cite{scalemine} samples matches of a pattern in the data
graph to build a statistical profile of pattern frequency, which guides the
computation of exact results. IncGM+\cite{inc-fsm} is a system for continuous
frequent subgraph mining that prunes the search space using a set of infrequent
subgraphs which are adjacent to frequent subgraphs. None of these works are
pattern-based, and instead view the FSM computation in terms of subgraphs of the
data graph.

\paragraph{Graph Processing.} 
Many works address computations on static and dynamic graphs~\cite{pregel,powergraph,galois,ligra,graphbolt,chaos,graphx,gemini,kickstarter,aspen,aspire,coral,gps,pgxd,input-reduction}.
Research in this space is typically applicable to computations on vertex
and edge values, as opposed to graph mining problems which are concerned with
substructures of the data graph. Techniques like 
\cite{input-reduction} develop custom transformations for specific substructures
to shrink the data graph and speed up value propagation.

\section{Conclusion}
We presented \textsc{Pattern Morphing} to enable structure-aware algebra over patterns. Pattern Morphing \emph{morphs} a given set of patterns into alternative pattern sets that can be used to compute accurate results for the original set of patterns. We developed the theoretical foundations for Pattern Morphing that guarantee accurate conversion of matches across different patterns, and further enable direct conversion of results from application-specific operations over explored matches. To evaluate the effectiveness of pattern morphing, we implemented an automatic pattern morphing engine and incorporated it in the Peregrine system. Our results showed that pattern morphing significantly improved the performance of various graph mining applications. Being a generic technique at pattern-level, pattern morphing can be incorporated in existing graph mining systems, and it can be applied to various problems beyond graph mining.

\balance
\bibliographystyle{plain}
\bibliography{paper}

\end{document}